\documentclass[journal]{IEEEtran}

\usepackage{graphicx,graphics,epsfig,times,bm,mathrsfs}
\usepackage[normalem]{ulem}
\usepackage{xfrac}
\usepackage{mathtools}
\usepackage{amsmath,amssymb,amsfonts,amsthm}
\usepackage[table, xcdraw]{xcolor}
\usepackage{varwidth}
\usepackage{babel}
\usepackage{subfigure}
\usepackage{hhline}
\usepackage{multirow}
\usepackage{tikz}
\usepackage{longtable}
\usepackage{pgf}
\usepackage{physics}
\usepackage{bbm}
\usepackage{qcircuit}
\usepackage{pgfplots}
\pgfplotsset{compat=1.17} 
\usepackage{listings}
\usepackage{paralist}
\usepackage[hidelinks]{hyperref}
\usepackage{float}
\usepackage[final]{microtype}

\usepackage[T1]{fontenc}
\usepackage[utf8]{inputenc}
\usepackage{algorithm}
\usepackage[noend]{algpseudocode}
\usepackage{tikzsymbols}
\usepackage{paralist}
\usepackage[font=small,skip=3pt]{caption}

\tikzset{>=latex}
\usetikzlibrary{
    external,
    patterns,
    arrows,
    positioning,
    chains,
    decorations.pathmorphing,
    decorations.markings,
    decorations.pathreplacing,
    shapes,
    shadows.blur,
    shapes.symbols,
    calligraphy,
    fadings,
    calc,
}

\makeatletter

\newcounter{protocol}

\newtheorem{theorem}{Theorem}

\newenvironment{protocol}[1][htb]{%
  \let\c@algorithm\c@protocol
  \renewcommand{\ALG@name}{Protocol}
  \begin{algorithm}[#1]%
  }{\end{algorithm}
}
\makeatother

\newcommand{\enquote}[1]{``#1''}

\author{
\IEEEauthorblockN{Ankit Khandelwal, Stephen DiAdamo \vspace{-8mm}}
\thanks{© 2023 IEEE. Personal use of this material is permitted. Permission from
IEEE must be obtained for all other uses, in any current or future media,
including reprinting/republishing this material for advertising or promotional
purposes, creating new collective works, for resale or redistribution to servers
or lists, or reuse of any copyrighted component of this work in other works.}
}

\title{Enhancing Protocol Privacy with Blind Calibration of Quantum Devices}

\begin{document}

\maketitle

\begin{abstract}

To mitigate the noise in quantum channels, calibration is used to tune the devices to minimize error. Generally, calibration is performed by transmitting pre-agreed-upon calibration states and determining an error cost so the two parties can tune their devices accordingly. The calibration states can be the same ones used for the desired protocol, and so an untrusted party could potentially learn which protocol is being performed by gathering knowledge of the calibration states and cost function. Here, we assume privacy of the protocol is the goal and therefore the receiver should not be allowed to determine the protocol states. We limit the information that is revealed to the receiver, and in this regard, we propose a simple protocol that hides the calibration states and cost function from the receiver, but still allows for calibration to be performed efficiently, thereby increasing the privacy of the protocol. We show various numerical results demonstrating the ability of the protocol under various channel noise parameters and communication scenarios.

\end{abstract}

\vspace{-3mm}
\section{Introduction}

Quantum network technologies are rapidly developing and the potential applications are becoming more realistic for commercial deployment. A key benefit quantum networks bring is an unreachable level of privacy and security that classical networks cannot currently achieve, with protocols such as quantum key distribution \cite{pirandola2020advances}, secret sharing \cite{hillery1999quantum}, and blind quantum computing \cite{fitzsimons2017private}. A similarity that each of these protocols share is that, before the protocols begins, and periodically during a continuous execution, the communicating parties perform device calibration to ensure their devices are aligned, minimizing any noise-induced errors caused by the quantum channel. Optical components are affected by environmental changes like temperature fluctuations~\cite{yeung1978effect}, and so it is necessary to regularly calibrate the devices to mitigate these effects.

By performing calibration, it opens the network up to potential attack. For example, in the security proofs of some QKD protocols, it is assumed that the devices begin in a calibrated state, but it has been shown that there are attacks that can be performed during the calibration protocol \cite{jain2011device, fei2018quantum}, thereby diminishing the security. Further, in QKD it is critical that an accurate estimate for the Quantum Bit Error Rate (QBER) is known so that information reconciliation can be accurately and efficiently performed \cite{elkouss2010information}, and so calibration is essential. It is therefore important for the calibration stage to be as secure as the protocols themselves, not only for QKD, but in all quantum secure communication scenarios. 

An example approach to device calibration for quantum communication is to align the transmitter and receiver using polarized states as described in \cite{xavier2009experimental}. Here the sender and receiver tune a polarization controller to calibrate, where the sender and receiver both know the calibration states. The process repeats until the two parties minimize the error rate. In this work, we propose a calibration strategy that uses the same quantum signals that will be used in the protocol to be executed, allowing calibration, but done in a way such that the receiver cannot determine the states, thereby protecting the secrecy of the protocol. 

In this work, we propose a simple approach to overcoming the need for the sender to reveal the calibration states and the cost function used for calibration. We call this approach \enquote{blind calibration} because from the perspective of the receiver, they are blind to which quantum states they are receiving and why the tuning of their device is improving transmission cost. As far as we know, this is the first proposal for blind calibration used for quantum devices. In sensor networks \cite{balzano2007blind} blind calibration was proposed as a calibration method without accurate ground-truth, and so a different meaning of ``blind''. The goal here is to hide information intentionally from the receiver. The approach here is to randomize the transmission order of the calibration states, thereby preventing the receiver from performing quantum state tomography to determine the state set. Moreover, the sender hides the cost function by computing it at the their side of the channel, sending only a---potentially  obfuscated---value back to the receiver. The sender's goal is therefore to complete the calibration as efficiently as possible, while hiding as much information as possible. In this work, we propose a communication protocol that achieves this goal and we analyze various transmission scenarios for proof-of-concept. With simulation, we test how well our approach works over fiber channel models with various noise parameters. To gain advantage using our protocol, although, calibration should be done infrequently enough to outweigh its cost since our protocol is potentially more resource intensive. 

Throughout this work we use the following mathematical notation: the set of quantum states on a finite dimensional, complex Hilbert space  $\mathcal{H}$ of dimension $n$ is denoted $\mathcal{S}(\mathcal{H})$; a quantum channel is a Completely Positive Trace Preserving (CPTP) map over the state space $\mathcal{S}(\mathcal{H})$; random variables are usually denoted as $X_i$ and vectors of random variables with a bold-face notation $\mathbf{X}_i$.

\section{Blind Calibration over Quantum Channels}\label{sec:protocol}

In this section, we introduce our protocol for blind calibration. In order to specify the protocol, we need to firstly state the assumptions we make. Similar to the main assumptions for many quantum cryptographic protocols, we assume the following: 1) Both parties act to achieve the goal of the protocol and do not act to prevent it from being fulfilled; 2) The sender has a trusted and private random number generator; 3) The sender operates in a secure location and particularly their quantum source has no side channels leaking information. Further, we rely on the fact that quantum states cannot be perfectly copied, that is, the no-cloning theorem.

With this, we state our protocol for blind calibration. In Protocol~\ref{proto:blind-calibration}, we list the the protocol instructions for the sender and receiver. The high level idea of the protocol is that the sender obfuscates the transmission statistics to prevent the receiver from performing tomography \cite{nielsen_chuang_2010}, but at the same time, using efficient optimization algorithms to find the tuning parameters that minimize the desired cost function. To do this, for each iteration of the protocol, the sender uses the uniform distribution $p$ over a set of calibration states $S$ to determine the order to transmit $N$ quantum states. From there, $N$ transmissions are made to the receiver after being processed by an encoder, and the receiver then performs a decoding operation and measures in a particular basis. The measurement results are sent back to the sender with the basis information, and a cost of transmission is determined by the sender. The cost is used to update the encoding strategy and is further sent to the receiver to update the decoding strategy. An updating process tunes the transmission and detection devices to perform calibration. This process repeats until convergence or a maximum number of iterations $I_{\max}$ is reached. In this sense, any eavesdropper intercepting the quantum states in transmission cannot perform a tomography attack to determine which quantum states are being transmitted. 

We further justify why tomography at the receiver alone cannot be performed on the transmitted states to uncover the calibration state set $S$. The original algorithm for tomography is as follows. For a state $\rho \in \mathcal{S(H)}$ from a Hilbert space $\mathcal{H}$ of $n$ dimensions, we have the following equation,
\begin{equation}\label{eq:tomo_multi}
    \rho=\frac{1}{n}\sum_{\Vec{v}}\Tr(\sigma_{v_1}\otimes\dots\otimes\sigma_{v_n}\rho)\sigma_{v_1}\otimes\dots\otimes\sigma_{v_n}
\end{equation}

Here, the vectors $\Vec{v}=(v_1,\dots,v_n)$ with entries $v_i$ are taken from the set $\{I,X,Y,Z\}$ of Pauli matrices. 
In order to reconstruct a density matrix with tomography, one needs to make many measurements in the different basis such that the state statistics converge accordingly. Usually, it is assumed that the state on which tomography is being performed is the same for each measurement. In our protocol, we ensure that the receiver is unable to determine when, from a set of calibration states, a specific quantum state was transmitted because each state is transmitted uniformly randomly. The best the receiver can do is compute an average state matrix, specifically a maximally mixed state, which cannot reveal state information of the parts. We prove the claim formally.

\begin{figure}
\begin{protocol}[H]
\raggedright
\caption{\raggedright{Blind Calibration}}
\small{
\textbf{Parameters}
\begin{compactitem}
    \item $S \subseteq \mathcal{S}(\mathcal{H})$, a set of calibration states, $|S|>0$
    \item $N\in\mathbbm{N}$, the number of transmissions per iteration
    \item $Cost:\mathbbm{R}^N\times\mathbbm{N}^N\rightarrow\mathbbm{R}$, a cost function
    \item $\epsilon_{\text{th}} \geq 0$, a calibration convergence threshold
    \item $I_{\text{max}}$, the maximum number of iterations
    \item Parameterized Encoder $E_\theta:\mathcal{S}(\mathcal{H}_A)\rightarrow\mathcal{S}(\mathcal{H}'_A)$
    \item Quantum Channel $C:\mathcal{S}(\mathcal{H}'_A)\rightarrow\mathcal{S}(\mathcal{H}'_B)$
    \item Parameterized Decoder $D_\phi:\mathcal{S}(\mathcal{H}'_B)\rightarrow\mathcal{S}(\mathcal{H}_B)$
    \item $Update:\mathbbm{R}^M\times \mathbbm{R}\rightarrow\mathbbm{R}^M$, a parameter update function
\end{compactitem}}
\textbf{Sender}
\begin{algorithmic}[1]
     \State $i \gets 0$; $\theta^i\gets$ Initialize
     \While{$i<I_\text{max}$ and threshold $\epsilon_{\text{th}}$ not reached}
        \State $statesTransmitted \gets [ \hspace{1mm} ]$
        \For{$t \in \{0,..., N\}$}
        \State \textbf{choose} State $
        \rho \in S$ with uniform probability
        \State \textbf{append} Index of $
        \rho$ to $statesTransmitted$ 
        \State \textbf{transmit} $E_{\theta^i}(\rho)$ over channel $C$
        \EndFor
     \State $M\gets$ \textbf{await} Measurement results
     \State $cost \gets Cost(M, statesTransmitted)$
     \State \textbf{transmit} $cost$ classically to Receiver
     \State $\theta^{i+1} \gets Update(\theta^i, cost)$; $i\gets i+1$
     \EndWhile
    \State Terminate protocol
\end{algorithmic}
\textbf{Receiver}
\begin{algorithmic}[1]
     \State $i \gets 0$; $\phi^i\gets$ Initialize
     \While{protocol not terminated}
     \State $M\gets [\hspace{1mm}]$
     \For{$t \in \{0,..., N\}$}
     \State $ \tilde{\rho} \gets$ \textbf{await} Quantum channel input
     \State $m \gets$  \textbf{measure} $D_{\phi^i}(\tilde{\rho})$ 
     \State \textbf{append} $m$ to $M$ with basis information
     \EndFor
     \State \textbf{transmit} $M$ classically to sender
     \State $cost \gets$ \textbf{await} Cost value decreased
     \State $\phi^{i+1}\gets Update(\phi^i, cost)$; $i\gets i+1$
     \EndWhile
\end{algorithmic}
\label{proto:blind-calibration}
\end{protocol}
\vspace{-8mm}
\end{figure}

\begin{theorem}
    Given a state set $S_p = \{\rho_1, ..., \rho_i\} \subseteq \mathcal{S}(\mathcal{H})$, where $\dim(\mathcal{H}) = n$,  and an $\epsilon>0$, there exists calibration set $S \supseteq S_p$, and $N$, the number of transmissions, such that, using Protocol \ref{proto:blind-calibration}, a receiver following the protocol using quantum state tomography can only recover a state set $\tilde{S}_p$ such that for each $\tilde{\rho} \in \tilde{S}_p$, $\norm{I/n - \tilde{\rho}} \leq \epsilon$.
\end{theorem}

\begin{proof}
    Let $S_p = \{\rho_1, ..., \rho_k\} \subseteq \mathcal{S(H)}$ be a state set and $\epsilon>0$ be given. We can construct $S_p$ in the following way. For each state $\rho_j \in S_p$, $S$ will contain both $\rho_j$ and $\overline{\rho_j} = U(\Vec{\pi})\rho_jU^\dagger(\Vec{\pi})$, that is, the state from the state set and a state with the exact opposite behavior, where $U$ is a rotation operator in $n$ dimensions. Let $p$ be the uniform distribution over $S$. For each transmission, the sender sends a state $\rho_t \in S$ with probability $p(\rho_t)=1/|S|$. For state tomography, the values $\Tr\left(\bigotimes_{i=1}^n\sigma_{v_i} \rho\right)$ are estimated via,
    \begin{align}
        \frac{1}{N}\sum_{i=1}^N m_i = \frac{1}{N}\sum_{i=1}^n n_i m_i \overset{\makebox[0pt]{\mbox{\normalfont\footnotesize\sffamily $\epsilon/n$}}}{\approx}
        \frac{1}{n}\sum_{i=1}^n m_i,
    \end{align}
    where $n_i$ represents the number of times measurement outcome $m_i$ was observed. Since either a state or its opposite is transmitted with equal probability, in each basis, measuring one value is as probable as measuring another, and so with the use of the law of large numbers, each $n_i$ approaches $N/n$ for large values of $N$, where such an $N$ can always be found to be $\epsilon/n$ away from the true distribution. What results is therefore a near equal mixture of all measurement outcomes. If we assume the measurement outcomes $m_i$ are normalized to a $[-1, 1]$ line, the sum  \eqref{eq:tomo_multi} becomes $\Tr(\rho)I/n \pm \epsilon/nI =I/n \pm \epsilon/nI$.
\end{proof}

\section{System Model and Simulation Configuration}\label{sec:model}
In this section, we present a model for a communication system that performs Protocol \ref{proto:blind-calibration}. The components of the system are a quantum source, an encoding device, the noisy quantum channel, a decoding device, and a measurement device. To model these components, we use the following. For simplicity of analysis, we assume that the quantum source emits quantum states with perfect fidelity, and therefore there is no noise model associated with it. For the encoder and decoder, we assume that they are quantum channels parameterized by a vector of variables $\theta^i = (\theta^i_{0}, \theta^i_{1}, ..., \theta^i_{n})$ for the encoder and  $\phi^i = (\phi_{0}^i, \phi_{1}^i, ..., \phi_{m}^i)$ for the decoder, where the superscript $i$ represents the current iteration of the protocol. For the measurement device, we again assume for simplicity that measurement outcomes are accurate.

For the quantum channel, we consider noise models for coherent noise. Firstly we model random noise by applying small, random, unitary rotations, and secondly we consider length-dependent bit- and phase-flip errors. The noise models are applied via quantum channels such that each qubit in the system has an independent chance of experiencing the noise.

For the rotational noise model, we generate a random rotation vector $\mathbf{\theta}_R$ with elements from $-\pi$ to $\pi$. Then, depending on the length $L$ kilometers of the channel, we determine the percentage of the rotation applied via, $p = 1 - 10^{-\mu L/10}$, where $\mu$ is a parameter of the channel. For the length-dependent noise, for a channel of length $L$ and channel parameters $\mu_1, \mu_2 \geq 0$, the probability that a qubit arrives with a bit-flip error or a phase-flip error is $p_b = 1 - 10^{-\mu_1 L/10}$ and $p_p = 1 - 10^{-\mu_2 L/10}$ respectively. 

For the protocol, we need to select a cost function. Here any valid cost function can work, but in our simulations, we use the quantum fidelity or error rates. Fidelity between states is defined as follows. Assuming the input state is $\sigma \in \mathcal{S}(\mathcal{H})$ and the output state is $\rho\in \mathcal{S}(\mathcal{H})$, the state fidelity between the two is given by, $F(\rho,\sigma) \coloneqq \Tr\left(\sqrt{\sqrt{\rho}\sigma\sqrt{\rho}}\right)^2.$ To use fidelity as a cost, we define infidelity to be $1 - F$. To determine the fidelity in simulation, we implemented quantum state tomography \cite{nielsen_chuang_2010} as explained in the previous section. To compute the fidelity nonlocally in our case, the receiver needs to perform the measurements and the sender needs to perform the state reconstruction to compute the state fidelities of the calibration states. This process is much the same as the calibration method in \cite{xavier2009experimental}, it differs only in that the measurement outcomes are sent back to the sender for state reconstruction instead being  computed locally.

The error rates as a cost function are calculated by comparing the classical input to the measured output. For an input vector $in$ of classical messages and output $out$, the error rate is computed as, $err(in, out) \coloneqq 1 - \frac{1}{n}\sum_{i=1}^n \delta_i(in, out)$, where $\delta_i(a, b) = 1$ if $a_i = b_i$ and $0$ otherwise and $n=|in|=|out|$.

To numerically analyze Protocol~\ref{proto:blind-calibration}, we use a collection of simulated experiments developed using the quantum computing framework PennyLane \cite{bergholm2018pennylane}. We configure our simulations as follows. We begin by generating a parameterized quantum circuit in PennyLane that performs the encoding step, the quantum channel, the decoding step, and the measurements as in Fig.~\ref{fig:ent-swap-circuit}. For updating the encoder and decoder parameters, after each iteration, updated values are passed to the circuit. For the update function, we use an optimizer to minimize the cost function. The encoder and decoder parameters are then updated during protocol execution according to the optimizer outputs, and in the next iteration, noise compensation values are used. For noise parameters, depending on the length of the channel, we provide the circuit a set of parameters for the rotational and probabilistic bit- or phase-flip noise values. The parameters of the protocol and the set of calibration states are chosen and then simulation can run. 

\section{Protocol Performance Benchmarks}\label{sec:performance}
We benchmark the performance of Protocol~\ref{proto:blind-calibration} for various communication scenarios under various noise parameters. To simplify the simulation model, we assume that the sender simply prepares quantum sources but does not modify their encoding during the protocol. In our notation, this implies $E_{\theta^i}$ is the identity map for all iterations $i$. Indeed, with tuning ability, a sender can pre-calibrate their device locally, ensuring the states they intend to transmit are calibrated to their standard, and so generality is not lost with this assumption. 

\subsection{Random Single-Qubit States}

As an initial test of the protocol, we analyze the case where the sender randomly selects $n=5$ single-qubit states for transmission. During transmission, the receiver measures the incoming states in one of the three Pauli bases such that quantum state tomography can be performed at the sender's side. For a cost function, the sender computes an average state fidelity using the measurements from the receiver to compare against the originally chosen states.

For this test of the protocol, we fix the length of the fiber to 50 km and determine how many iterations of the protocol are needed for approximate convergence to determine a good parameter regime for a more in-depth study of the protocol that we conduct in the later sections. The number of states transmitted per iteration here is $N=15,000$. For this analysis, we use a simple rotational noise model which rotates qubits as defined above with $\mu=0.05$ dB/km. Bit- and phase-flip noise is also set to $\mu_1=\mu_2=0.05$ 
dB/km.  We then deterministically apply a Bloch sphere rotation $p\cdot \mathbf{\theta}_R$ to each qubit that travels through the channel. For flip errors, they are applied randomly with their respective probability. For decoding, we apply a counter rotational correction before a measurement is made. In Fig.~\ref{fig:fid-vs-iters}, we can see that with approximately 100 iterations of the protocol, the protocol has learned how to overcome the rotational noise effects in the channel. Based on these results, for the remainder of this work, we use $I_{\max}=250$ iterations for our simulations unless otherwise stated.

\begin{figure}
    \centering
    \includegraphics[]{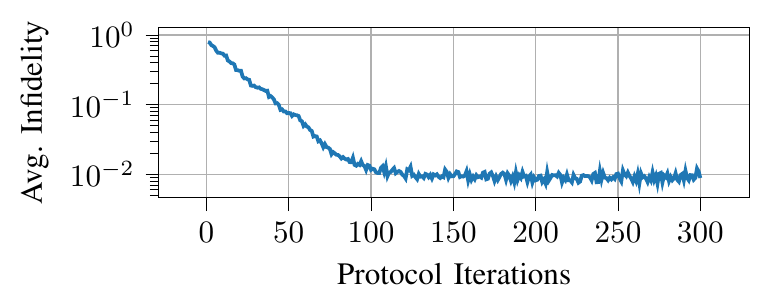}
    \caption{The average infidelity vs. number of protocol iterations for calibration using $n=5$ randomly selected states over 50 km fiber under deterministic rotational noise.\vspace{-3mm}}
    \label{fig:fid-vs-iters}
\end{figure}

\subsection{BB84 States}\label{sec:bb84}

The BB84 protocol uses states $\{\ket{0}, \ket{1}, \ket{+}, \ket{-}\}$ and transmits them with equal probability. A naive calibration process is to tell the receiver the basis in which the state is prepared and repeatedly transmit the state. The receiver can then tune the decoder according to a quantum bit error rate and this can be repeated for all states. The receiver then knows exactly which states are the protocol states and what the cost function is, and with this information, the receiver (or an eavesdropper imitating the receiver), can potentially predict the protocol the sender intends to perform. Using Protocol~\ref{proto:blind-calibration} instead, the sender need not disclose the calibration states nor the cost function, thereby masking the protocol parameters from the receiver, preventing the receiver from learning which protocol the sender intends to perform. 

The QKD results after performing Protocol~\ref{proto:blind-calibration} for a channel with only rotational noise are plotted in Fig.~\ref{fig:qber-length} (a). We can see that as the fiber distance increases, the QBER will tend to 0.6 in the pre-calibrated stage (blue). Using Protocol~\ref{proto:blind-calibration} to calibrate the system, the trained decoder reduces the BB84 QBER in the system within 250 iterations of training with 1,000 (lossless) transmissions per iteration. In Fig.~\ref{fig:qber-length} (b) is the performance over a channel that has both rotational noise and probabilistic bit- and phase-flip error. What we can see is that until roughly 60 km, calibration makes no difference to accommodate the noise in the channel, since the probability of the error applying is smaller than 0.5. Once the flipping probability exceeds 0.5, the QBER can be minimized by applying the inverse rotations deterministically, which is observed in the bifurcation of the trends at around 60 km. 

\begin{figure}
    \centering
    \includegraphics[]{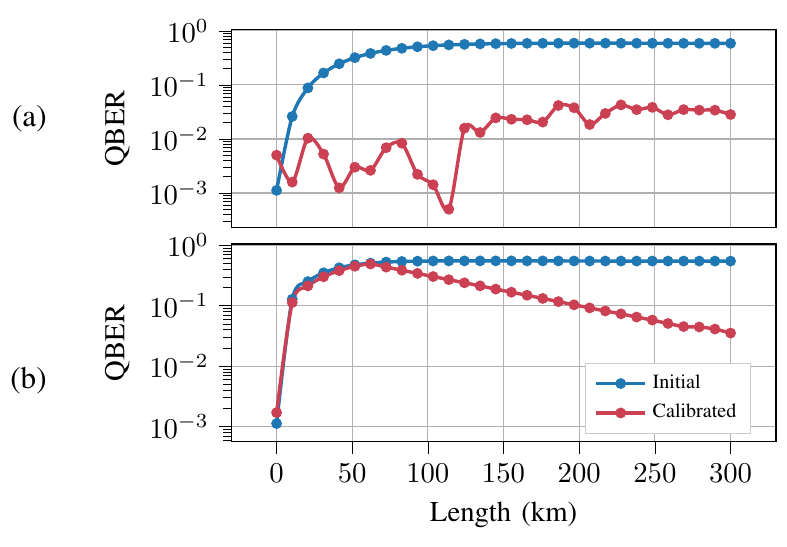}
    \caption{Calibration result for mitigating (a) rotational noise, and (b) rotation, bit- and phase-flip noise optimized with $\epsilon_\text{th} = 10^{-7}$. Plotted is the average Quantum Bit Error Rate (QBER) for 350 trials of performing the BB84 protocol with 1,000 qubits sent per trial.\vspace{-4mm}}
    \label{fig:qber-length}
\end{figure}

We test how many transmissions are required to minimize the QBER within 20 iterations. The number of transmissions build up accuracy in the statistics regarding the noise properties of the channel. The results appear in Fig.~\ref{fig:shots-qber}. In the case of a channel with both rotation noise and bit- and phase-flip error with a fixed channel length of $L=120$ km, at approximately $N=2,000$, further  transmissions do not improve the calibration significantly. Here, we do not consider loss in the system, so in this case we can consider $N$ to be the number of single qubit detections. In a channel with attenuation $\eta\in (0, 1]$, we can expect the number of transmissions to increase by a factor of $\eta^{-1}$. 

\begin{figure}
    \centering
    \includegraphics[]{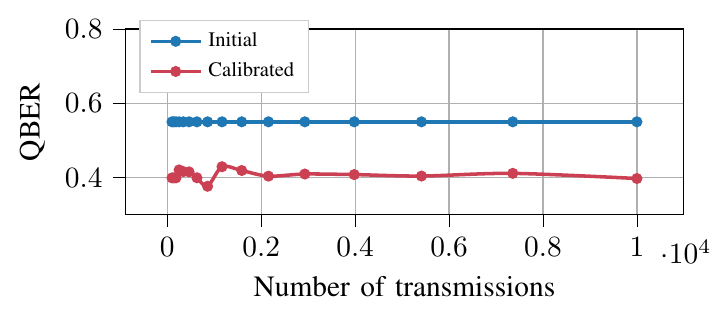}    
    \caption{The calibration accuracy using a fixed $I_{\max}=20$ iterations over a noisy channel with both bit- and phase-flip and rotational error.}
    \label{fig:shots-qber}
\end{figure}

\subsection{Entanglement-Swapping}

In quantum networks, an important protocol for extending the range of quantum state transmission is entanglement swapping~\cite{zukowski1993event}. To perform entanglement swapping, three (or more) parties, a sender, a receiver, and a midpoint receiver need to coordinate to perform the protocol. One approach to this protocol is that the sender and receiver each locally generate an EPR pair and send half of the EPR pair to the midpoint node. The midpoint node performs a Bell State Measurement (BSM) and at the end of the protocol, the sender and receiver share an EPR pair. This configuration is also seen in QKD protocols such as Measurement Device Independent QKD (MDI-QKD)~\cite{lo2012measurement}, making this use-case a further practical scenario. 

\begin{figure}
    \centering
    \includegraphics[]{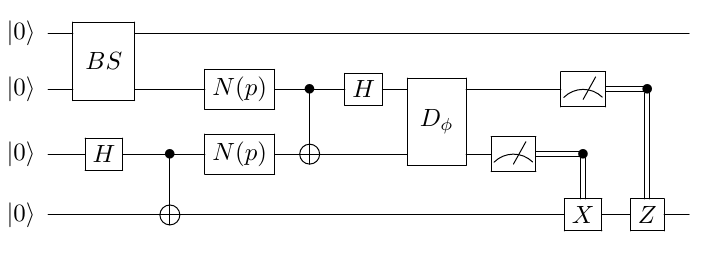}
    \caption{Entanglement swap circuit. The $BS$ gate prepares one of four Bell states. $N(p)$ gates apply noise. \vspace{-4mm}}
    \label{fig:ent-swap-circuit}
\end{figure}

To implement the protocol, we use the following approach based on  the circuit in Fig.~\ref{fig:ent-swap-circuit}. We assume that the both the sender and receiver transmit half of a Bell pair to a central node, and the sender selects a Bell state at random. The central node performs a correction to overcome channel noise before measuring, but is unaware of which Bell state is being distributed. The central node performs a Bell state measurement on the incoming qubits. To teleport the quantum state, two classical bits used for correction are sent to one party. After this, the sender and receiver parties measure their remaining half and can coordinate with each other to generate a cost function, using the error cost defined previously for example. where the input is Bell states that were transmitted, and the output is the resultant measured Bell pair.

We perform Protocol~\ref{proto:blind-calibration} over two error models. In Fig.~\ref{fig:ent-swap} (a), we analyze the effect of arbitrary rotation. In the Fig.~\ref{fig:ent-swap} (b), we include bit- and phase- flip error. With deterministic errors, as are the arbitrary rotations, it is possible to correct for error. For probabilistic error, we see that until the probability is high enough for error, no correction using rotational correction can be effectively applied. Once it is the case, we see the calibrated error rates diverge from the uncorrected error rates, as described in the QKD example. Here again we do not simulate loss, and a scaling factor should be applied to accommodate loss as explained in Section~\ref{sec:bb84} for each channel and for detector flaws. 

\begin{figure}
    \centering
    \includegraphics[]{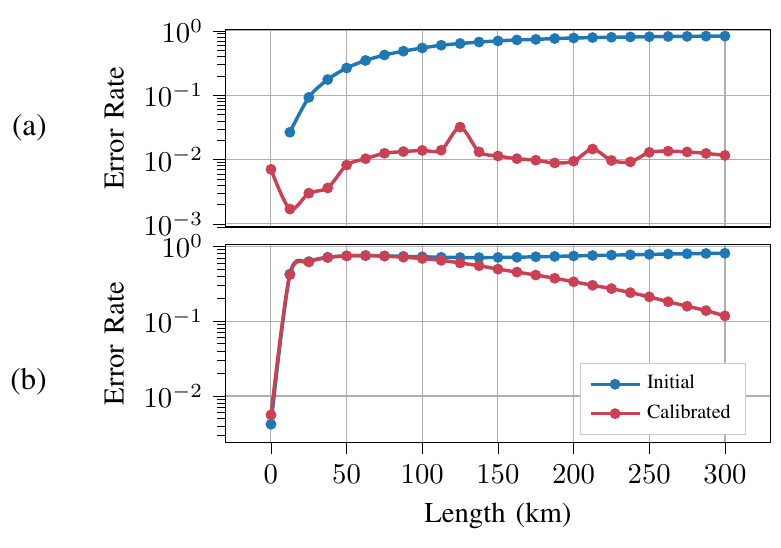}
    \caption{Calibration performance for entanglement swapping under (a) rotational noise and (b) rotational and bit-flip noise optimized with a cost-change tolerance of $\epsilon_\text{th} = 10^{-7}$. Plotted is the average error rate for transmitting performing Bell state transmissions. \vspace{-2mm}}
    \label{fig:ent-swap}
\end{figure}

\vspace{-3mm}

\subsection{Multipartite Entanglement Distribution}

In this section, we simulate a communication setting for transmitting multipartite entanglement, specifically GHZ and W states, over channels with rotational noise. GHZ states and W states are defined as, $\ket{GHZ} \coloneqq  \frac{1}{\sqrt2}(\ket{00...0} + \ket{11...1})$ and $\ket{W} \coloneqq \frac{1 }{\sqrt{n}}(\ket{100...0} + \ket{010...0} + ... + \ket{00...01})$. GHZ states and W states are important for protocols like quantum secret sharing \cite{hillery1999quantum} and others. Here we analyze a setting where a sender prepares $n$ qubit GHZ or W states and transmits them over a noisy channel. 

In Fig.~\ref{fig:multi-partitle}, we show the calibration effects for correcting the channel noise for both a varying number of qubits in the entangled quantum systems and length of the channel. With a growing number of qubits in the system, there is a higher rate of infidelity, but as with the previous results, we see that the protocol can robustly correct deterministic errors in the channel. Here we have used the state infidelity as the cost, approximating the output state density matrix using multi-qubit state tomography in Eq.~\eqref{eq:tomo_multi}. For better performance of the simulation, we use $\epsilon_{\text{th}}=10^{-5}$ and 20,000 channel transmissions per protocol iteration. Here we do not simulate probabilistic noise, as simulation performance rapidly deteriorates with larger entangled states, but given the noise models are the same as the previous cases, we suspect a similar trend would emerge, where until roughly $L=60$ km no corrections can be made. 

\begin{figure}
    \centering
    \includegraphics[]{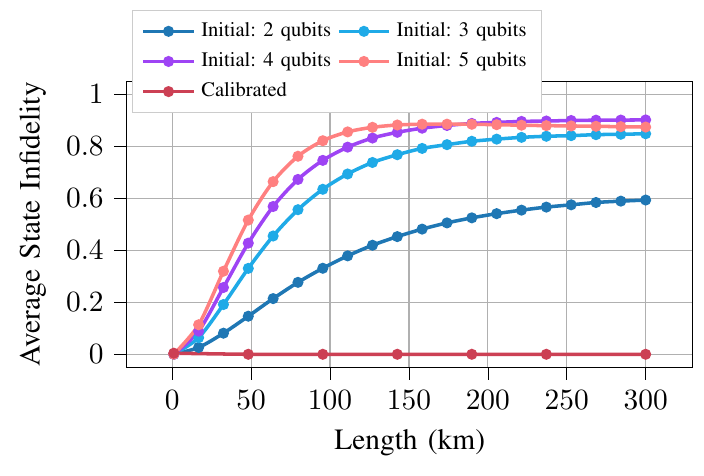}
    \caption{Calibration performance for transmitting $n=2,3,4,5$ qubit GHZ and W states over a channel with rotational noise. Here we optimize with a cost-change tolerance of $\epsilon_\text{th} = 10^{-5}$, and use 20,000 transmissions per iteration. The calibrated trend is for all $n$ qubit cases. \vspace{-3mm}}
    \label{fig:multi-partitle}
\end{figure}

\vspace{-3mm}

\section{Conclusion}\label{sec:conclusion}

In conclusion, we show that masking the calibration states and cost while still allowing the calibration to take place is possible. To do this, our Blind Calibration protocol is used, preventing any other party from performing state tomography to learn the calibration state set. By performing the error cost calculation at the sender's side, there is further masking of the calibration protocol, minimizing any leakage of the calibration protocol. We show that, under various communication settings, the protocol works well even under noisy channel conditions. In our analysis, we used simple decoding strategies to overcome noise, but more complex strategies can directly fit with the protocol framework. The protocol we propose is general and accommodates any quantum communication setting where calibration is required making it practical and timely during the growing interest in quantum communication and quantum networks.

\vspace{-3mm}

\bibliographystyle{IEEEtran}

\end{document}